\documentclass[letterpaper, 10 pt, conference]{ieeeconf}
\usepackage[utf8]{inputenc}

\usepackage{amssymb,amsmath,amsthm}
\usepackage[dvipsnames]{xcolor}
\usepackage{soul}
\usepackage{graphicx}
\usepackage{caption}
\usepackage{subcaption}
\usepackage[left=2cm,right=2cm,top=2cm,bottom=2cm]{geometry}
\usepackage{xcolor}

\newtheorem*{proposition*}{Proposition}
\newtheorem*{theorem*}{Theorem}
\newtheorem{assumption}{Assumption}

\newtheorem{theorem}{Theorem}
\newtheorem{definition}{Definition}
\newtheorem{remark}{Remark}

\newcommand{\EE}{\mathbb{E}}
\newcommand{\Fs}{\mathcal{F}}

\newcommand{\Ns}{\mathcal{N}}
\newcommand{\Qs}{\mathcal{Q}}

\newcommand{\Gs}{\mathcal{G}}

\newcommand{\Vs}{\mathcal{V}}

\newcommand{\RR}{\mathbb{R}}

\DeclareMathOperator*{\argmin}{argmin}

\title{\LARGE \bf Adversarial Linear-Quadratic Mean-Field Games over Multigraphs} 

\author{Muhammad~Aneeq~uz~Zaman, Sujay~Bhatt, and Tamer~Ba{\c s}ar
	\thanks{First and third authors are affiliated with the Coordinated Science Laboratory, University of Illinois at Urbana–Champaign Urbana, IL 61801. Emails:  mazaman2@illinois.edu, sujaybhatt.hr@gmail.com, basar1@illinois.edu}
	\thanks{Research of first and third authors supported in part by an AFOSR Grant (FA9550-19-1-0353), and in part by an ARO MURI Grant (AG285).}
}

\IEEEoverridecommandlockouts
\begin{document}
	\maketitle
	
	\begin{abstract}
		In this paper, we propose a game between an exogenous adversary and a network of agents connected via a multigraph. The multigraph is composed of (1) a global graph structure, capturing the virtual interactions among the agents, and (2) a local graph structure, capturing physical/local interactions among the agents. The aim of each agent is to achieve consensus with the other agents in a decentralized manner by minimizing a local cost associated with its local graph and a global cost associated with the global graph. The exogenous adversary, on the other hand, aims to maximize the average cost incurred by all agents in the multigraph. We derive Nash equilibrium policies for the agents and the adversary in the Mean-Field Game setting,  when the agent population in the global graph is arbitrarily large and the ``homogeneous mixing" hypothesis holds on local graphs. This equilibrium is shown to be unique and the equilibrium Markov policies for each agent depend on the local state of the agent, as well as the influences on the agent by the local and global mean fields.

	\end{abstract}
	
	\section{Introduction}

Mean-field games (MFGs) model large scale strategic interactions in a network of individual rational entities (a.k.a agents), where each agent tries to optimize its individual objective function. Because of the presence of a very large number of agents in a MFG setting, the strategic interaction is such that the impact of any one agent on the other agents is negligible; however, the overall effect of the network on each agent cannot be ignored. There is rich literature on MFGs in both continuous time and discrete-time (\cite{huang2006large,lasry2007mean,moon2014discrete}), with applications in diverse fields, ranging from cybersecurity to inter-bank borrowing/lending (\cite{uz2020secure, carmona2013mean}).

In this paper, we introduce Adversarial MFGs (A-MFGs) as a class of MFGs where an exogenous adversary and a network of agents strategically interact to realize self-objectives. The objective of each agent is to achieve consensus with other agents in a decentralized manner, while the objective of the exogenous adversary is to adversely affect the cumulative objectives of the agents in the network. One area of application of A-MFGs is social networks, and particularly study of evolution of opinions in adversarial settings.

Heretofore, the influence of the structure of the network on the mean-field interactions has not been  extensively studied in the literature. In this paper, we consider MFGs on \textit{multigraphs} with an exogenous adversary, where we explicitly include the effect of the network on the multi-agent interactions. Multigraphs (\cite{gjoka2011multigraph,shafie2015multigraph}) are graphs representing networks where multiple edges standing for distinct relationships between the agents are permitted. As a representative example, see Fig. \ref{fig:multigraph} where the dark edges represent physical/local interactions and light edges represent virtual/social interactions.

Specifically, we consider linear dynamics for the agents and incorporate the effect of interaction over the multigraph edges via a local cost associated with its local neighborhood and the global cost associated with the collective behavior of the network. The exogenous adversary is able to influence the agents through the agents' dynamics and cost functions, and aims to maximize the average cost over all agents. By considering the MFG (limiting case where the number of agents tends to infinity), it is possible to propose approximate-Nash equilibrium strategies for the agents and the adversary in the finite population game. 

\begin{figure}[htbp]
	\centering
	\includegraphics[scale=.25]{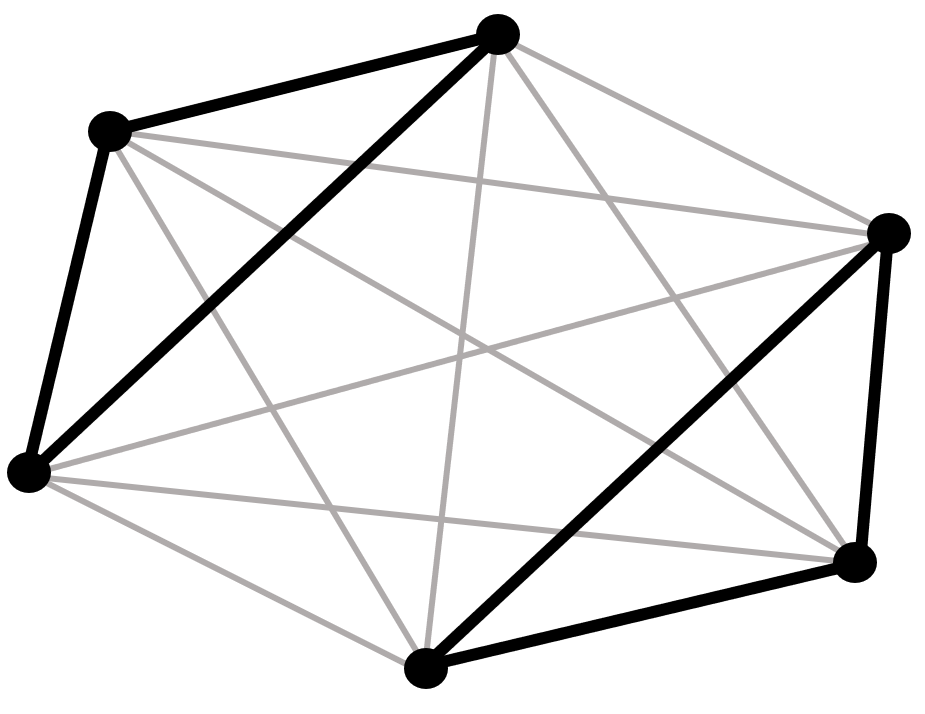}
	\caption{\small The multigraph is composed of two underlying graphs. The global graph (light edges) represents the social/virtual interactions among agents and strongly connects all the agents in the multigraph. The local graph (dark edges) represents physical/inter-personal interactions and connects an agent to its neighborhood.}
	\label{fig:multigraph}
\end{figure}

\subsection{Motivating example of adversarial MFGs on multigraphs} \label{subsec:mot_exmpl}
The framework and results of this paper can be used to analyze the evolution of opinion dynamics on multigraphs. Consider a multigraph in which the dark edges in Fig. \ref{fig:multigraph} provide a backbone for the opinion evolution via word-of-mouth and the light edges act as a backbone for the opinion evolution via virtual/social media. While with global interactions (light edges) consensus can be reached, with local interactions (dark edges) only clusters of consensus opinions may emerge \cite{shrimplin2011contradictions}. However, we would be remiss not to consider the role of exogenous media houses acting to manipulate the opinion evolution \cite{baum2008relationships}. 

\subsection{Main results \& organization}
\noindent The contributions of this paper are as follows:
\begin{enumerate}
\item We formalize the Adversarial MFG (A-MFG) over multigraphs to model the evolution of opinion dynamics over a network of agents connected through virtual/global and physical/local networks
\item We establish the existence and uniqueness of the  mean-field equilibrium for Linear Quadratic Mean Field Games (LQ-MFGs) on multigraphs. Furthermore, we show that the local and global effects display some useful linearity properties, which can be leveraged to obtain equilibrium control policies for the agents and the adversary. This extends the results available in the literature by explicitly considering the network structure in the MFG framework.
\item We show that, for an A-MFG on a multigraph, the equilibrium strategy for a generic agent depends linearly on the local state of the agent, the local effect (effect of the neighborhood on the agent) and the global effect (aggregate effect of all the agents).
\end{enumerate}

\subsection{Related literature}
Mean-field games (MFGs) were introduced in the seminal works of \cite{huang2006large,huang2003individual} and, independently, \cite{lasry2007mean}, to address the issue of scalability in dynamic games. Rooted in the original MFG formulation, LQ-MFGs have been studied in the continuous-time setting (\cite{huang2007large,bensoussan2016linear,huang2018linear,bardi2012explicit}) as well as the discrete-time setting (\cite{moon2014discrete,uz2020approximate,fu2019actor}). There is also growing interest in using Reinforcement Learning for MFGs in the setting where model parameters are unknown to the agents (\cite{subramanian2019reinforcement,fu2019actor,uz2020reinforcement}). 

MFGs with a network structure were first considered in \cite{delarue2017mean}, where the underlying graph was of the Erd\"{o}s–R\'{e}nyi type. Recently there has been a number of papers on Graphon MFGs (\cite{caines2018graphon,gao2020lqg,caines2019graphon}) which utilizes graphon theory to model infinite limits of large networks. Our work here appears to be the first to investigate the local/physical vs global/virtual graph dichotomy which emerges in social networks. Similar to our setting, literature on MFGs with dominating \cite{bensoussan2016mean} or major agents (\cite{lasry2018mean,ma2020linear,huang2010large}) considers one agent which has a significant effect on the other players. On the contrary the significant player (adversary) explicitly aims is to minimize the average objective of the other agents and has no idiosyncratic noise in this work. Works on risk-sensitive and robust MFGs (\cite{tembine2013risk,bauso2016robust,moon2016linear,bauso2016opinion}) bear semblance to this work as they consider the presence of an adversarial player through risk-sensitivity or robustness considerations but differ on account of the multigraph network structure of this work. 

The paper is organized as follows. Section \ref{sec:formulation} formulates the A-MFG over multigraphs, and discusses the relevant solution concepts in finite population and infinite population settings. Section \ref{sec:mfe_char} deals with characterization of the equilibrium of the game in the infinite population setting. Section \ref{sec:conc} provides concluding remarks and describes future directions of research. The proofs are included in the appendix.

	\section{Problem Formulation} \label{sec:formulation}


In this section we first introduce the $N$-agent adversarial game over multigraphs. The game is composed of an exogenous adversary strategically interacting with $N$ rational agents connected with each other through a multigraph. Due to the decentralized information structure solving the finite population game becomes difficult. So we formulate the A-MFG over multigraphs, where the exogenous adversary interacts with infinitely many players connected via a multigraph.

\subsection{Adversarial LQ game over multigraphs}
We propose a nonzero-sum Linear Quadratic game, where an exogenous adversary interacts with $N$ agents and all the agents are connected through an underlying multigraph structure. The agents are connected with each other through two distinct graphs (global and local), hence the name multigraph. The global graph structure represents strategic interactions of a global nature, like virtual/social media interactions, as shown by the gray edged graph in~Fig.~\ref{fig:multigraph}. The local graph structure represents the physical/real-world interactions among subsets of agents, like neighborhoods or friendships as shown by the dark edged graphs in Fig.~\ref{fig:multigraph}. A set of agents connected by the local graph (dark) is called a \textit{neighborhood}. In our formulation, neighborhoods do not overlap, and hence each agent $i$ belongs to one and only one neighborhood, which is denoted by $\Ns(i)$. 

The state of each agent $i$ at time $t$ is denoted by $Z^i_t \in \RR^z$ and its control effort by $U^i_t \in \RR^u$. The adversary's control is denoted by $V_t \in \RR^v$, through which it can affect the dynamics and the cost functions of the agents in the multigraph. Each agent (say $i$) has linear dynamics, where the agent's control $U^i_t$ and the adversary's control $V_t$ enter linearly, as given below.
\begin{align} \label{eq:fin_pop_dyn}
Z_{t+1}^i & = A Z_t^i +B U_t^i + C V_t + W_t^i, t \in (0,\ldots,T-1)
\end{align} 
where $W_t^i \sim \Gs(0,\Sigma_w)$ denotes an i.i.d. Gaussian process noise of agent $i$ and $Z^i_0 \sim \Gs(\mu_0, \Sigma_0)$ is the Gaussian distributed initial state of agent $i$. The initial states $Z^i_0$ and the noise processes $W_t^i$ of all the agents are assumed to be independent of each other. The aggregate state of a neighborhood $\Ns(i)$ at time $t$ is the denoted by $Y_t^i$ where 
\begin{align} \label{eq:fin_pop_dyn_net}
Y_t^i = \sum_{j \in \mathcal{N}(i)} Z_t^j / \lvert \mathcal{N}(i) \rvert 
\end{align} 
In this work, we assume \emph{homogeneous mixing} \cite{turnes2014epidemic} of the local graph, which says that the neighbors of an agent $i$ can be chosen independently from the set of all $N$ agents. This will be expressed precisely for the mean-field setting where the number of agents is assumed to be large. The individual cost function each agent wants to minimize is:
\begin{align} \label{eq:fin_pop_cost_net}
&J^i_N (\pi^i,\pi^{-i},V) = \sum_{t=0}^{T-1} \EE\bigg[ \big\lVert Z_t^i \big\rVert^2_{Q_t}  + \bigg\lVert Z_t^i- \frac{1}{N} \sum_{j = 1}^{N}Z_t^j \bigg\rVert^2_{\bar{Q}_t} \nonumber \\
& + \big\lVert Z_t^i- Y_t^i \big\rVert^2_{\tilde{Q}_t} + \big\lVert U_t^i \big\rVert^2_{R_t} - \big\lVert V_t \big\rVert^2_{S_t} \bigg] + \EE \bigg[ \big\lVert Z_T^i \big\rVert^2_{Q_T} \nonumber \\
&  + \bigg\lVert Z_T^i- \frac{1}{N} \sum_{j = 1}^{N}Z_T^j \bigg\rVert^2_{\bar{Q}_T} + \big\lVert Z_T^i- Y_T^i \big\rVert^2_{\tilde{Q}_T} \bigg] .
\end{align}

Each agent $i$ at time $t$ is assumed to have access to only its local history $Z^i_{[0,t]} = (Z^i_0,\dots,Z^i_t)$. 
For agent $i$, the set of all such policies $\pi^i := (\pi^i_0, \ldots, \pi^i_{T-1})$ is denoted by $\Pi^i$. The agent aims to find the control policy $\pi^i$ which minimizes its cost \eqref{eq:fin_pop_cost_net}.

Note that the agent's cost function, as in \eqref{eq:fin_pop_cost_net}, is composed of a local term $\big\lVert Z_t^i- Y_t^i \big\rVert^2_{\tilde{Q}_t}$ due to the interaction in the local graph and a global term $  \big\lVert Z_t^i- \frac{1}{N} \sum_{j = 1}^{N}Z_t^j \big\rVert^2_{\bar{Q}_t}$ due to the interaction in the global graph. These terms are consensus type terms and penalize the deviation of agent $i$ from the local behavior and global aggregate behavior, respectively. Hence, the agent aims to align itself with local as well as global consensus. The cost function also includes terms that penalize large values of the state of the agent and its control over the horizon, as well as (negatively) the control effort of the adversary. The weighting matrices $Q_t,\bar{Q}_t,\tilde{Q}_t,S_t$ and $R_t$ in \eqref{eq:fin_pop_cost_net} are symmetric matrices such that $Q_t,\bar{Q}_t, \tilde{Q}_t \geq 0, S_t,R_t > 0$ for $t \in [T]$.

The adversary is exogenous to the multigraph and can affect all the agents through their dynamics \eqref{eq:fin_pop_dyn} and cost functions \eqref{eq:fin_pop_cost_net}. The adversary uses its influence over the multigraph to disrupt consensus between agents by maximizing the average cost of all the agents, that is minimize
\begin{align} \label{eq:adv_cost}
J^0_N (V, \pi^{(N)}) = -\frac{1}{N} \sum_{i=1}^{N} J^i_N (\pi^i,\pi^{-i},V)
\end{align}
where $\pi^{(N)} = (\pi^1,\ldots,\pi^N)$ is the joint policy. The adversary at time $t$, is assumed to have access to the history of states of all the agents up to time $t$, $Z_{[0,t]}$. The set of all possible adversary policies under this information structure, is denoted by $\Vs$.

The solution concept used for this non-cooperative setting is that of Nash equilibrium which is the set of policies such that if all the agents and the adversary follow these policies, then none of the agents or the adversary have any incentive to unilaterally deviate. This is formally introduced as follows.
\begin{definition}
	The set of agent policies $\pi^* = (\pi^{1*},\ldots,\pi^{*N})$ and adversary policy $V^* = (V^*_0, \ldots, V^*_{T-1})$ constitute a Nash equilibrium if, for $i = 1,\ldots,N$,
	\begin{align*}
	J^i_N(\pi^{i*},\pi^{-i*},V^*) &  \leq J^i_N(\pi^{i},\pi^{-i*},V^*), \hspace{0.2cm} \pi^i \in \Pi^i,  \\
	J^0_N(\pi^*, V^*) &  \leq J^0_N(\pi^*, V), \hspace{0.2cm}V \in \Vs
	\end{align*}
\end{definition}
Due to the local decentralized information structure, computing Nash equilibria in the finite population setting might prove to be difficult \cite{bacsar1998dynamic}, so we turn to use of the MFG setting where $N \rightarrow \infty$. We call this the Adversarial Mean-Field Game (A-MFG) over a multigraph, which is different from the standard LQ-MFGs in two distinct ways: 1) there is an adversary which aims to manipulate the agents to maximize their average cost, and 2) each agent is also affected by a local interaction term, in addition to the global interaction term, due to the presence of local and global graphs in the multigraph. We assume the local effect follows the homogeneous mixing hypothesis, as to be further clarified in the following subsection.
\subsection{Adversarial MFG over multigraphs}
In the mean-field setting of this game, we consider a generic (representative) agent interacting with infinitely many agents, where the global aggregate is denoted by $\bar{Z}_t$ and the local aggregate is denoted by $Y_t$. The total number of agents $N \rightarrow \infty$, but the number of agents in a neighborhood may still be finite. The dynamics of the generic agent in the MFG is given by
\begin{align} \label{eq:gen_agent_dyn}
Z_{t+1} & = A Z_t +B U_t + C V_t + W_t.
\end{align}
We dropped the superscript for simplicity. The cost function the generic agent aims to minimize is given by
\begin{align} \label{eq:gen_agent_cost}
& J(\mu,\bar{Z},V) =  \sum_{t=0}^{T-1} \EE\big[ \big\lVert Z_t \big\rVert^2_{Q_t} \hspace{-0.1cm} + \big\lVert Z_t- Y_t \big\rVert^2_{\tilde{Q}_t} \hspace{-0.1cm} + \big\lVert Z_t- \bar{Z}_t \big\rVert^2_{\bar{Q}_t} \nonumber \\
& + \big\lVert U_t \big\rVert^2_{R_t} - \big\lVert V_t \big\rVert^2_{S_t} \big] + \EE \big[ \big\lVert Z_T \big\rVert^2_{Q_T} + \lVert Z_T- \bar{Z}_T \rVert^2_{\bar{Q}_T} \nonumber \\
& + \big\lVert Z_T- Y_T \big\rVert^2_{\tilde{Q}_T} \big]
\end{align}
where $\bar{Z} = (\bar{Z}_0, \ldots, \bar{Z}_T)$ represents the global aggregate (also called global mean-field), $Y = (Y_0, \ldots, Y_{T})$ represents the local aggregate (local mean-field) and adversary control policy is denoted by $V = (V_0,\ldots,V_T)$. In this setting, the quantities $\bar{Z}$ and $V$ are assumed to be deterministic signals. The control $\mu_t$ of the generic agent at time $t$, is adapted to the filtration generated by the state of the agent and the local mean-field, $\Fs^Z_t$ and $\Fs^Y_T$ respectively. We assume that the multigraph follows the homogeneous mixing property and as a result the local mean-field is assumed to satisfy the following property.
\begin{assumption} {\bf [Homogeneous Mixing]}
	The local mean-field $Y_t$ is assumed to be an exogenous noise process with mean $\EE[Y_t] = \bar{Z}_t$. 
\end{assumption}
\begin{remark}
This assumption is justified by the homogeneous mixing hypothesis \cite{turnes2014epidemic}, which states that an agent's neighbors are randomly distributed among all agents in the game. This hypothesis is quite standard in infection \cite{turnes2014epidemic,del2013mathematical} and opinion \cite{baumgaertner2018spatial} dynamics literature and is being used in mean-field games for the first time to the best of the author's knowledge. We notice that this hypothesis is valid only for large population games.
\end{remark}

The adversary aims to disrupt consensus between the agents, by maximizing the average cost of the generic agent, given the control policy $\mu$ of the generic agent and the mean-field trajectory,
\begin{align*}
J^0 (V,\mu, \bar{Z}) = \lim_{N \rightarrow \infty} J^0_N (V,\mu^{(N)})
\end{align*}
where $\mu^{(N)} = (\mu,\ldots,\mu)$ signifies that all the agents (in the finite population setting) follows the control policy $\mu$. The mean-field equilibrium is an analog of the Nash equilibrium for the MFG and is defined as the $4$-tuple $(\mu,V,\bar{Z},Y)$, where $\mu = (\mu_0,\ldots,\mu_{T-1})$ is the control policy of the generic agent, $V = (V_0,\ldots,V_{T-1})$ is the control policy of the adversary, $\bar{Z}$ is the global mean-field and $Y$ is the local mean-field of the neighborhood of the generic agent. It is formally defined as follows.
\begin{definition}
	The tuple $(\mu^*,V^*,\bar{Z}^*, Y^*)$ is a mean-field equilibrium if {\bf (1) Optimality}:
	\begin{align*}
	\mu^* = \argmin_{\mu} J(\mu,\bar{Z}^*,V^*), \hspace{0.2cm} V^* = \argmin_{V} J^0 (V,\mu^*, \bar{Z}^*)
	\end{align*}
	and {\bf (2) Consistency}: $\bar{Z}^*$ is the aggregate behavior of the infinitely many agents and $Y^*$ is the aggregate behavior of the neighborhood of the generic agent, if the agents follow control policy $\mu^*$ and adversary follows control policy $V^*$.
\end{definition}

The quantities $Y^*$ and $\bar{Z}^*$ are called the \emph{equilibrium} global and local mean-fields (MFs), respectively and the control policies $V^*$ and $\mu^*$ are called the \emph{equilibrium} adversarial and generic agent control policies, respectively. Next we characterize the mean-field equilibrium (MFE) of the A-MFG over multigraphs, by using open-loop analysis. This is accomplished by first developing the maximum principle for stochastic tracking. This maximum principle is then used to obtain the form of optimal control policies for the generic agent and the adversary. Then the equilibrium local and global MFs are characterized under these optimal control policies and the consistency conditions.
	
	\section{MFE Characterization} \label{sec:mfe_char}
	In this section we characterize the MFE of A-MFG over a multigraph. We start by providing the form of equilibrium control policy of the generic agent and the adversary in section \ref{subsec:gen_cont}. This is obtained by developing a stochastic maximum principle (SMP) for tracking a stochastic reference signal. The behavior of the agents and the adversary under equilibrium control policies is aggregated to characterize the equilibrium global MF in section \ref{subsec:gmf_anl_adv_pol}. This also leads to a closed-form expression for adversary's equilibrium policy. Finally we characterize the behavior of the equilibrium local MF and the equilibrium policy of the generic agent. 

\subsection{Equilibrium control policies} \label{subsec:gen_cont}
In this subsection, we obtain the form of equilibrium control policies for the generic agent and the adversary. First we deal with the generic agent's equilibrium control problem for a given deterministic global MF $\bar{Z} = (\bar{Z}_0, \ldots, \bar{Z}_{T-1})$, adversary control policy $V = (V_0, \ldots, V_{T-1})$, and an exogenous stochastic local MF noise process $Y = (Y_0, \ldots, Y_{T-1})$. The control of the generic agent is assumed to be adapted to filtrations generated by its state and the local MF, $\Fs^Z$ and $\Fs^Y$, respectively. 
Recalling the dynamics of the generic agent,
\begin{align}
Z_{t+1} = A Z_t + B U_t + C V_t + W_t, Z_0 \sim \Gs (\mu_0, \Sigma_0) ,
\end{align}
and that the cost function of the generic agent is coupled with the other agents through the global MF $\bar{Z}_t$ and the local MF $Y_t$ in the following manner,
\begin{align*}
& J =  \sum_{t=0}^{T-1} \EE\big[ \big\lVert Z_t \big\rVert^2_{Q_t} \hspace{-0.15cm} + \big\lVert Z_t- Y_t \big\rVert^2_{\tilde{Q}_t} \hspace{-0.15cm} + \big\lVert Z_t- \bar{Z}_t \big\rVert^2_{\bar{Q}_t} \hspace{-0.15cm} + \big\lVert U_t \big\rVert^2_{R_t} \nonumber \\
&  - \big\lVert V_t \big\rVert^2_{S_t} \big] \hspace{-0.1cm} + \hspace{-0.05cm} \EE \big[ \big\lVert Z_T \big\rVert^2_{Q_T} \hspace{-0.1cm} + \lVert Z_T- \bar{Z}_T \rVert^2_{\bar{Q}_T} \hspace{-0.1cm} + \big\lVert Z_T- Y_T \big\rVert^2_{\tilde{Q}_T} \big]
\end{align*}
where $\bar Z_t$ and $V_t$ are deterministic signals and $\mu$ is adapted to the state process $Z_t$ and \emph{exogenous} local MF $Y_t$. We obtain the SMP for this problem by using a technique similar to \cite{chau2017discrete}. This maximum principle is novel as it solves the problem of linear quadratic tracking of a stochastic signal.
\begin{theorem} \label{thm:agent_cont}
	The generic agent's equilibrium control adapted to filtration $\Fs^Z_t \vee \Fs^{Y}_t$ is given by
	\begin{align*}
	U^*_t = - R^{-1}_t B^T \zeta_{t+1}
	\end{align*}
	where $\zeta_t$ can be constructed as
	\begin{align*}
	&\zeta_t = A^T \zeta_{t+1} + (Q_t + \bar{Q}_t + \tilde{Q}_t) Z_t - \bar{Q}_t \bar{Z}_t - \tilde{Q}_t Y_t -   M^\zeta_t, \nonumber \\
	& \zeta_T = (Q_T + \bar{Q}_T + \tilde{Q}_T) Z_T - \bar{Q}_T \bar{Z}_T - \tilde{Q}_T Y_T, \\
	& M^\zeta_t = A^T \zeta_{t+1} - A^T \EE \big[\zeta_{t+1} \mid \Fs^Z_t \vee \Fs^{Y}_t \big].
	\end{align*}
	where $  M^\zeta_t$ is a martingale difference sequence adapted to filtration $\Fs^Z_t \vee \Fs^{Y}_t$.
\end{theorem}
The agent's equilibrim policy doesn't depend explicitly on the adversary's control actions but might have implicit dependence through the mean-field trajectories. Notice that the co-state equation depends on a martingale difference sequence adapted to a filtration. This makes the control policy adapted to the filtration $\Fs^Z_t \vee \Fs^{Y}_t$. This gives us the form of equilibrium control of the generic agent but still does not provide its closed-form expression. Towards that end we need to first characterize the equilibrium local and global MFs. 
And for this, we first derive the form of equilibrium control of the adversary, which is shown to depend on just the global MF of the game.
\begin{theorem} \label{thm:adv_cont}
	If the following condition is satisfied
	\begin{align} \label{eq:valid_cond}
	S_t - C^T \hat{P}_{t+1} C > 0 ,
	\end{align}
	where the matrix $\hat{P}_{t}$ is defined recursively by
	$\hat{P}_t = Q_t + A^T \hat{P}_{t+1} A + 
	A^T \hat{P}_{t+1} C (S_t - C^T \hat{P}_{t+1} C)^{-1} C^T \hat{P}_{t+1} A, \hspace{0.2cm} \hat{P}_T = Q_T $, 
	then the equilibrium control policy of the adversary is,
	\begin{align} \label{eq:adv_cont}
	V^*_t = - S^{-1}_t  C^T  \bar{\zeta}^0_{t+1}, \bar{\zeta}^0_t = A^T  \bar{\zeta}^0_{t+1} - Q_t \bar{Z}_t, \bar{\zeta}^0_T = - Q_T \bar{Z}_T ,
	\end{align}
	where  $\bar{\zeta}^0_t$ is the adversary's co-state and $\bar{Z}$ is the global MF of the agents.
\end{theorem}
The theorem provides the form of equilibrium control for the adversary. The proof starts by obtaining a SMP for the adversary's equilibrium control problem in the finite population setting. Then by taking the limit $N \rightarrow \infty$ we arrive at the conclusion. The condition on positive definiteness of the matrix $S_t - C^T \hat{P}_{t+1} C$ is standard condition in two-player zero sum games \cite{bacsar1998dynamic}. This arises because in essence the adversary plays a zero sum game with the global MF. 

\subsection{Equilibrium MF analysis and equilibrium policies} \label{subsec:gmf_anl_adv_pol}
Having obtained the form of equilibrium control for both the adversary and the generic agent, we now characterize the equilibrium local and global MFs. This will allow us to compute the closed-form expressions for the equilibrium control policies of the generic agent and the adversary. MF characterization involves proving existence, uniqueness and some useful properties of the equilibrium global and local MFs. 
\begin{theorem} \label{thm:eq_glo_MF}
	If condition \eqref{eq:valid_cond} is satisfied, then equilibrium global MF follows linear dynamics,
	\begin{align} \label{eq:gmf_lin}
	\bar{Z}^*_{t+1} = E^{-1}_t A \bar{Z}^*_t = \bar{F}_t \bar{Z}^*_t,
	\end{align}
	where $E_t = (I + B R^{-1}_t B^T \bar{P}_{t+1} - C S^{-1}_t C^T \bar{P}_{t+1}) $ and $\bar{P}_t$ is given by the Riccati equation, 
	$\bar{P}_t = A^T \bar{P}_{t+1} E^{-1}_t A + Q_t, P_T = Q_T$, 
	and the equilibrium adversarial policy,
	\begin{align} \label{eq:adv_cont_fin}
	V^*_t = S^{-1}_t C^T \bar{P}_t \bar{F}_t \bar{Z}^*_t .
	\end{align}
	is linear in the equilibrium global MF.
\end{theorem}
Now we need to characterize the dynamics of equilibrium local MF so as to obtain the closed-form expression of the generic agent's equilibrium control policy. Let us define $ Y^*_t = \sum_{j \in \Ns} Z^{j*}_t/\lvert \Ns \rvert$, where each $j$ corresponds to a generic agent chosen independently from infinitely many agents (due to homogenous mixing property) and $\Ns$ corresponds to the neighborhood of the generic agent. The dynamics of the equilibrium local MF can be characterized as follows,
\begin{theorem} \label{thm:eq_loc_MF}
	If the condition \eqref{eq:valid_cond} is satisfied, then local equilibrium MF has linear Gaussian dynamics driven by the equilibrium global MF:
	\begin{align} \label{eq:lmf_lin}
	Y^*_{t+1} = \tilde F^1_t  Y^*_t   + \tilde F^2_t \bar{Z}^*_t + \tilde  E^{-1}_t \tilde W_t 
	\end{align}
	where $\tilde W_t = \sum_{j \in \Ns} W^j_t/\lvert \Ns \rvert$, 
	\begin{align*}
	\tilde E_t = & (I + B R^{-1} B^T \tilde P_{t+1} - C S^{-1}_t C^T \tilde P_{t+1}), \tilde F^1_t = \tilde E^{-1}_t A, \\
	\tilde F^2_t = & \tilde E^{-1}_t (B R^{-1}_t B^T - C S^{-1}_t C^T))  \sum_{i=0}^{T-t} \prod_{j=1}^{i} \tilde H_{t+j} \bar{Q}_{t+i}  \bar{F}_{t+j} \nonumber \\
	\tilde P_t = & A^T \tilde P_{t+1}\tilde E^{-1}_t A + Q_t + \bar{Q}_t, \hspace{0.2cm} \tilde P_T = Q_T + \bar{Q}_T,
	\end{align*}
	and $\tilde H_t = A^T (I - \tilde P_{t+1} \tilde E^{-1}_t B R^{-1} B^T)$.
\end{theorem}
Having shown that the equilibrium global MF follows deterministic linear dynamics and the equilibrium local MF follows stochastic linear dynamics driven by the equilibrium global MF, we now turn our attention to obtaining the expression for the equilibrium control policy for the generic agent $\mu^*$. We start by concatenating the equilibrium local and global MFs with the state of the generic agent $X_t = [Z^T_t, Y^{*T}_t, \bar{Z}^{*T}_t]^T$. Now finding the generic agent's equilibrium control policy, turns into an LQG problem with dynamics,
\begin{align*}
X_{t+1} = \bar{A}_t X_t + \bar{B} U_t + \bar{W}_t
\end{align*}
where,
\begin{align*}
\bar{A}_t = \begin{pmatrix}
A & 0 & 0 \\ 0 & \tilde F^1_t & \tilde F^2_t \\ 0 & 0 & \bar F_t
\end{pmatrix}, \bar{B} = \begin{pmatrix}
B \\ 0 \\ 0
\end{pmatrix}, \bar{W}_t = \begin{pmatrix}
W_t \\ \tilde  E^{-1}_t \tilde W_t  \\ 0
\end{pmatrix}
\end{align*}
The reformulated cost function is
\begin{align*}
& J =  \sum_{t=0}^{T-1} \EE\big[ \big\lVert X_t \big\rVert^2_{\Qs_t} + \big\lVert U_t \big\rVert^2_{R_t} - \big\lVert V_t \big\rVert^2_{S_t} \big] + \EE \big[ \big\lVert Z_T \big\rVert^2_{Q_T} \big]
\end{align*}
where $V_t$ is given by \eqref{eq:adv_cont_fin} and 
\begin{align*}
\Qs_t = \begin{pmatrix}
Q_t + \bar{Q}_t + \tilde{Q}_t & - \tilde{Q}_t & - \bar{Q}_t \\
- \tilde{Q}_t & \tilde{Q}_t & 0 \\
- \bar{Q}_t & 0 & - \bar{Q}_t
\end{pmatrix}
\end{align*}
for $t = 0,\ldots, T$. If condition \eqref{eq:valid_cond} is satisfied, this is a standard time varying linear quadratic regulator (TV-LQR) problem. Using standard techniques \cite{lewis2012optimal} the equilibrium control for the generic agent is,
\begin{align*}
U^*_t = \mu^*_t (Z_t, Y^{*}_t, \bar{Z}^{*}_t) = -(R_t + \bar{B}^T P^*_{t+1} \bar{B})^{-1} \bar{B}^T P^*_{t+1} \bar{A}_t X_t
\end{align*}
where 
$P^*_t$ satisfies the Riccati equation:
$P^*_t = \Qs_t + \bar{A}^T_t P^*_{t+1} \bar{A}_t 
 - \bar{A}^T_t P^*_{t+1} \bar{B} (R_t + \bar{B}^T P^*_{t+1} \bar{B} )^{-1} \bar{B}^T P^*_{t+1} \bar{A}_t, P^*_T = \Qs_T$.

We have thus completely characterized the MFE of the A-MFG over networks. 
The equilibrium control policy of the adversary \eqref{eq:adv_cont_fin}, depends solely on the equilibrium global MF whose dynamics are deterministic and linear \eqref{eq:gmf_lin}. So the equilibrium global MF and equilibrium adversary policy can be precomputed. The equilibrium control policy of the generic agent, on the other hand, depends on the equilibrium local and global MFs. Since the dynamics of the equilibrium local MF is Gaussian \eqref{eq:lmf_lin}, it cannot be computed and hence needs to be observed by the generic agent. Thus the generic agent's equilibrium control policy has closed-loop dependence on the equilibrium local MF but open-loop dependence on the equilibrium global MF.

Considering opinion dynamics over multigraphs, this means that consensus over a global graph can be precomputed with high level of confidence (due to the large number of agents). But consensus over the local graph (neighborhood) must be actively measured as it is significantly perturbed by the deviations of individual agents in the neighborhood.

	
	\section{Conclusion} \label{sec:conc}
	We have introduced the model of the Adversarial Linear Quadratic Mean-Field Games over Multigraphs, to study the strategic interactions among a network of agents and an exogenous adversary. The network is of multigraph type which is composed of two graphs. The global graph (representing virtual interactions) connects together all the agents, whereas the local graph (representing physical interactions) couples a subset of agents to form a neighborhood. The agents aim to form consensus over the multigraph whereas the adversary aims to disrupt that consensus.

Under homogeneous mixing hypothesis of the multigraph, we have shown that the equilibrium control policy of the adversary depends linearly on the equilibrium global MF whereas the equilibrium control policy of the generic agent depends linearly on the equilibrium local and global MFs. Furthermore, we have also shown that the equilibrium global MF can be precomputed, as opposed to the equilibrium local MF which follows linear Gaussian dynamics. As a result the equilibrium control policy of the adversary can be computed offline whereas the equilibrium control policy of the generic agent has online dependence on the state of the generic agent, the equilibrium local and global MFs.

As one extension, our goal is to show that the MFE of the A-MFG over multigraphs, is also an approximate Nash equilibrium for the $N$-player Adversarial game over multigraphs. Our goal, in particular, is to investigate the level of the approximation as a function of increasing neighborhood and global graph size ($\lvert \Ns \rvert$ and $N$ respectively). Other possible extensions to this work include generalizing the local graph structure to overlapping neighborhoods with sparse interconnections.
	
	\bibliographystyle{IEEEtran} 
	\bibliography{references,MARL_Springer_1,MARL_Springer_2,RL} 
	\newpage
	\section*{Appendix}
	
\begin{proof}[Proof of Theorem \ref{thm:agent_cont}]
	Consider perturbation of optimal control $U_t + \tau \tilde{U}_t$, $\tau \in \RR^+$ and the control $\tilde{U} = (\tilde{U}_0, \tilde{U}_1,\ldots)$ is adapted to filtration $\Fs^Z_t \vee \Fs^{Y}_t$. The original state becomes $Z_t + \tau \tilde{Z}_t$ where
	\begin{align*}
	\tilde{Z}_{t+1} = A \tilde{Z}_t + B \tilde{U}_t, \tilde{Z}_0 = 0.
	\end{align*}
	Due to optimality of control the necessary Euler condition must be satisfied:
	\begin{align*}
	& 0 =  \frac{d}{d\tau} J \bigg\rvert_{\tau = 0} = 2 \EE \bigg[ \langle U_0, \tilde U_0 \rangle_{R_t} + \sum_{t=0}^{T} \big( \langle Z_t, \tilde Z_t \rangle_{Q_t} \\ 
	& + \langle U_t, \tilde U_t \rangle_{R_t} + \langle Z_t - \bar{Z}_t, \tilde Z_t \rangle_{\bar{Q}_t} + \langle Z_t - Y_t, \tilde Z_t \rangle_{\tilde{Q}_t} \big) \nonumber \\
	& + \langle Z_T, \tilde Z_T \rangle_{Q_T} + \langle Z_T - \bar{Z}_T, \tilde Z_T \rangle_{\bar{Q}_T} + \langle Z_T - Y_T, \tilde Z_T \rangle_{\tilde{Q}_T} \bigg] \nonumber 
	\end{align*}
	Let us introduce the adjoint process:
	\begin{align*}
	\xi_t & = A^T \xi_{t+1} + (Q_t + \bar{Q}_t + \tilde{Q}_t) Z_t - \bar{Q}_t \bar{Z}_t - \tilde{Q}_t Y_t - \Delta M^\xi_t, \nonumber \\
	\xi_T & = (Q_T + \bar{Q}_T + \tilde{Q}_T) Z_T - \bar{Q}_T \bar{Z}_T - \tilde{Q}_T Y_T
	\end{align*}
	where,
	$\Delta M^\xi_t = A^T \xi_{t+1} - A^T \EE \big[\xi_{t+1} \mid \Fs_t \big]$. 
	Note that $A^T \EE[\xi_{t+1} \mid \Fs_t] + (Q + \bar{Q}_t + \tilde{Q}_t) Z_t - \bar{Q}_t \bar{Z}_t - \tilde{Q}_t Y_t = \xi_t \in \Fs_t$. Consider,
	\begin{align*}
	& \langle \tilde Z_{t+1}, \xi_{t+1} \rangle - \langle \tilde Z_{t}, p_{t}  \rangle \nonumber \\
	& = \langle \tilde Z_{t+1} - \tilde Z_{t}, \xi_{t+1} \rangle - \langle \tilde Z_{t}, \xi_{t+1} - p_{t} \rangle  \nonumber \\
	& = \langle A \tilde Z_t + B \tilde U_t - \tilde Z_{t}, \xi_{t+1} \rangle - \langle \tilde Z_{t}, \xi_{t+1} - A^T \xi_{t+1} \nonumber \\
	& \hspace{0.4cm} - (Q_t + \bar{Q}_t + \tilde{Q}_t) Z_t + \bar{Q}_t \bar{Z}_t + \tilde{Q}_t Y_t + \Delta M^\xi_t \rangle \nonumber \\
	& = \langle \tilde Z_t, -(Q + \bar{Q}_t + \tilde{Q}_t) Z_t + \bar{Q}_t \bar{Z}_t + \tilde{Q}_t Y_t \rangle + \langle B \tilde U_t, \xi_{t+1} \rangle \nonumber \\
	& \hspace{0.4cm} + \langle \tilde Z_t, \Delta M^\xi_t \rangle 
	\end{align*}
	Summing up for $t = 0$ to $T-1$ and taking expectation,
	\begin{align*}
	0 = &  \EE \big[ \langle \tilde Z_T, -(Q_T + \bar{Q}_T + \tilde{Q}_T) Z_T + \bar{Q}_T \bar{Z}_T + \tilde{Q}_T Y_T \rangle \nonumber \\
	& + \sum_{t=0}^{T-1} \langle \tilde Z_t, -(Q_t + \bar{Q}_t + \tilde{Q}_t) Z_t + \bar{Q}_t \bar{Z}_t + \tilde{Q}_t Y_t \rangle \nonumber \\
	& + \langle B \tilde U_t, \xi_{t+1} \rangle + \langle \tilde Z_t, \Delta M^\xi_t \rangle  \big]
	\end{align*}
	Using the Euler condition and noting that $U_t$ and $\tilde U_t$ are adapted to process $\Fs^Z_t \vee \Fs^{Y}_t$,
	\begin{align*}
	0 & = \EE \big[ \sum_{t=0}^{T-1} \langle R_t U_t + B^T \xi_{t+1}, \tilde U_t \rangle \big] \nonumber \\
	& = \EE \big[ \sum_{t=0}^{T-1} \langle R_t U_t + B^T \EE[\xi_{t+1} \mid \Fs^Z_t \vee \Fs^{Y}_t] , \tilde U_t \rangle \big].
	\end{align*}
	Since $\tilde U_t$ is arbitrary, we get the optimal control $U_t = - R^{-1}_t B^T \EE [\xi_{t+1} \mid \Fs^Z_t \vee \Fs^{Y}_t]$, $t = 0,1, \ldots, T-1$. Let $\zeta_t = \EE[\xi_t \mid \Fs^Z_t \vee \Fs^{Y}_t]$ then
	\begin{align*}
	\zeta_t & = \EE[\xi_t \mid \Fs^Z_t \vee \Fs^{Y}_t] = A^T \EE[\xi_{t+1} \mid \Fs^Z_t \vee \Fs^{Y}_t]  \\
	& + (Q_t + \bar{Q}_t + \tilde{Q}_t) Z_t - \bar{Q}_t \bar{Z}_t - \tilde{Q}_t Y_t \nonumber \\
	& = A^T \zeta_{t+1} + (Q_t + \bar{Q}_t + \tilde{Q}_t) Z_t - \bar{Q}_t \bar{Z}_t - \tilde{Q}_t Y_t - \Delta M^\zeta_t \nonumber
	\end{align*}
	where
	\begin{align*}
	& \Delta M^\zeta_t = A^T \zeta_{t+1} - A^T \EE [p_{k+1} \mid \Fs^Z_t \vee \Fs^{Y}_t] \nonumber \\
	& = A^T \zeta_{t+1} - A^T \EE [ \EE [p_{k+1} \mid \Fs^Z_{t+1} \vee \Fs^{\bar{Z}}_{t+1}] \mid \Fs^Z_t \vee \Fs^{Y}_t] \nonumber \\
	& = A^T \zeta_{t+1} - A^T \EE \big[\zeta_{t+1} \mid \Fs^Z_t \vee \Fs^{Y}_t \big].
	\end{align*}
	Therefore the optimal control is 
	$U_t = - R^{-1}_t B^T \zeta_{t+1}$. 
	Since the cost is convex, the optimal control is unique and the necessary condition for optimality is also sufficient.
\end{proof}

\begin{proof}[Proof of Theorem \ref{thm:adv_cont}]
	First we obtain the form of the optimal control for the adversary, for the finite population setting. Using techniques similar to the proof of Theorem \ref{thm:agent_cont}, we can obtain the form of optimal control of adversary:
	\begin{align*}
	0 = & S V_t + C^T \sum_{i=1}^{N} \zeta^{i,0}_{t+1} = S V_t + C^T  \zeta^0_{t+1} \nonumber \\
	& \implies V_t = - S^{-1}_t C^T \zeta^0_{t+1},
	\end{align*}
	and the co-state equations are
	\begin{align*}
	\zeta^{i,0}_t = & - \frac{1}{N} Q_t Z^i_t -  \frac{1}{N} \tilde{Q}_t \big(Z^i_t - Y^i_t \big) - \frac{1}{N} \bar{Q}_t\big(Z^i_t - \bar{Z}_t \big) \nonumber \\
	& + \frac{1}{N} A^T \zeta^{i,0}_{t+1} -   M^{\zeta,0}_t, \hspace{0.4cm} t = 1,\ldots,T-1 \\
	\zeta^{i,0}_T = & - \frac{1}{N} Q_T Z^i_T -  \frac{1}{N} \tilde{Q}_T \big(Z^i_T - Y^i_T \big) - \frac{1}{N} \bar{Q}_T \big(Z^i_T - \bar{Z}_T \big) \nonumber
	\end{align*}
	for $i = (1,\ldots, N)$, where $\zeta^0_t := \sum_{i = 1}^{N}\zeta^{i,0}_t$ and $M^{\zeta,0}_t$ is a Martingale difference sequence adapted to the filtration $\Fs^0_t = \Fs^{Z^1}_t \vee \Fs^{Z^1}_t \ldots \Fs^{Z^N}_t \vee \Fs^{Z^N}_t$.
	\begin{align*}
	M^{\zeta,0}_t = \frac{1}{N} A^T \zeta^{i,0}_{t+1} - \frac{1}{N} A^T  \EE [\zeta^{i,0}_{t+1} \mid \Fs^0_t] .
	\end{align*}
	For this necessary condition to be sufficient, the adversary's cost needs to be convex in $V$. This is satisfied by the  concavity condition in 2-player zero-sum games \cite{bacsar1998dynamic},
	$
	S_t - C^T \hat{P}_{t+1} C > 0
	$
	where the matrix $\hat{P}_{t}$ is defined recursively,
	\begin{align*}
	\hat{P}_t &  = Q_t + A^T \hat{P}_{t+1} A + \\
	& A^T \hat{P}_{t+1} C (S_t - C^T \hat{P}_{t+1} C)^{-1} C^T \hat{P}_{t+1} A, \hspace{0.2cm} \hat{P}_T = Q_T \nonumber
	\end{align*}
	
	Now we take the limit $N \rightarrow \infty$ to arrive at the MFE of the game. Let $\bar{Z}_t$ and $\bar{\zeta}^0_t$ denote the global MF and adversary co-state at time $t$ respectively, such that
	\begin{align*}
	\bar{Z}_t = \lim_{N \rightarrow \infty} \sum_{i=1}^{N} \frac{Z^i_t}{N}, \hspace{0.2cm} \bar{\zeta}^0_t = \lim_{N \rightarrow \infty} \sum_{i=1}^{N} \zeta^{i,0}_t. 
	\end{align*}
	Notice that the first equation assumes that the tracking signal $\bar{Z}_t$ is also the aggregate state trajectory. We can write down the dynamics of these quantities as follows,
	\begin{align} \label{eq:adv_eq}
	\bar{Z}_{t+1} = & A \bar Z_t +B \bar U_t + C V_t, \nonumber \\
	V_t = & - S^{-1}_t  C^T  \bar{\zeta}^0_{t+1}, 
	\hspace{0.2cm} \bar{\zeta}^0_t =  A^T  \bar{\zeta}^0_{t+1} - Q_t \bar{Z}_t,  
	\end{align}
	Hence, we arrive at the adversary's equilibrium control policy.
\end{proof}

\begin{proof}[Proof of Theorem \ref{thm:eq_glo_MF}]
	To characterize the equilibrium global MF we start by substituting the equilibrium control of the generic agent from Theorem \ref{thm:agent_cont} into the dynamics of the generic agent:
	\begin{align*}
	Z_{t+1} & = A Z_t - B U_t + C V_t + W_t, \nonumber \\
	U^*_t & = - R^{-1}_t B^T \zeta_{t+1}  
	\end{align*}
	where $\zeta_t$ and $M^\zeta_t$ are defined in Theorem \ref{thm:agent_cont}. By taking expectation and invoking the \emph{consistency conditions} $\bar{Z}^*_t = \EE [Z_t]$ and $\bar{Z}^*_t = \EE [Y^*_t]$, we get
	\begin{align} \label{eq:gen_eq}
	\bar{Z}^*_{t+1} & = A \bar{Z}^*_t + B \bar{U}^*_t + C V^*_t, 
	\hspace{0.1cm} \bar{U}^*_t = - R^{-1}_t B^T \bar{\zeta}_{t+1} \nonumber \\
	\bar{\zeta}_t & = A^T \bar{\zeta}_{t+1} + Q_t \bar{Z}^*_t, \hspace{0.2cm} \bar{\zeta}_T = Q_T \bar{Z}^*_T 
	\end{align}
	where $\bar{\zeta}_t = \EE [\zeta_t]$. Comparing equations \eqref{eq:adv_eq} and \eqref{eq:gen_eq} we can see that $\bar{\zeta}_t = - \bar{\zeta}^0_t$ so the equilibrium global MF and costate dynamics can be written down as
	\begin{align}
	\bar{Z}^*_{t+1} & = A \bar{Z}^*_t - (B R^{-1}_t B^T - C S^{-1}_t C^T) \bar{\zeta}_{t+1} , \label{eq:gmf_fwd} \\
	\bar{\zeta}_t & = A^T \bar{\zeta}_{t+1} + Q_t \bar{Z}^*_t, \hspace{0.2cm} \bar{\zeta}_T = Q_T \bar{Z}^*_T  \label{eq:gmf_bck}
	\end{align}
	Let us assume the form of costate as $\bar{\zeta}_t = \bar{P}_t \bar{Z}^*_t + \bar{s}_t$. Substituting into \eqref{eq:gmf_fwd}, we get
	\begin{align}
	& \bar{Z}^*_{t+1} \nonumber  \\
	& = A \bar{Z}^*_t - (B R^{-1}_t B^T - C S^{-1}_t C^T) (\bar{P}_{t+1} \bar{Z}^*_{t+1} + \bar{s}_{t+1}), \nonumber \\
	& = E^{-1}_t ( A \bar{Z}^*_t - (B R^{-1}_t B^T - C S^{-1}_t C^T) \bar{s}_{t+1} ) \label{eq:gmf_inter}
	\end{align}
	where $E_t = (I + B R^{-1}_t B^T \bar{P}_{t+1} - C S^{-1}_t C^T \bar{P}_{t+1}) $. Now substituting $\bar{\zeta}_t = \bar{P}_t \bar{Z}^*_t + \bar{s}_t$ into \eqref{eq:gmf_bck}, we get
	\begin{align*}
	\bar{P}_t \bar{Z}^*_t + \bar{s}_t & = A^T (\bar{P}_{t+1} \bar{Z}^*_{t+1} + \bar{s}_{t+1}) + Q_t \bar{Z}^*_t, \\
	\bar{P}_T \bar{Z}^*_T + \bar s_T & = Q_T \bar{Z}^*_T \nonumber
	\end{align*}
	Substituting \eqref{eq:gmf_inter} into this equation, we arrive at
	\begin{align} \label{eq:gmf_covar}
	& \bar{P}_t \bar{Z}^*_t + \bar{s}_t = A^T (\bar{P}_{t+1} E^{-1}_t ( A \bar{Z}^*_t  \\
	& \hspace{0.5cm} - (B R^{-1}_t B^T - C S^{-1}_t C^T) \bar{s}_{t+1} ) + \bar{s}_{t+1}) + Q_t \bar{Z}^*_t \nonumber
	\end{align}
	Comparing coefficients of $\bar{Z}^*_t$ we obtain the Riccati eq.:
	\begin{align*}
	\bar{P}_t = A^T \bar{P}_{t+1} E^{-1}_t A + Q_t, \hspace{0.2cm} P_T = Q_T
	\end{align*}
	The variable $\bar{s}_t$ can be recursively computed from \eqref{eq:gmf_covar}:
	\begin{align}
	\bar{s}_t & = - A^T (\bar{P} E^{-1}_t (B R^{-1}_t B^T - C S^{-1}_t C^T) + I )\bar{s}_{t+1}, \nonumber
	\end{align}
	and $\bar{s}_T = 0$. Solving these equations we get $\bar{s}_t \equiv 0$. Hence, from \eqref{eq:gmf_inter} we can deduce that the equilibrium global MF follows linear dynamics:
	\begin{align} 
	\bar{Z}^*_{t+1} = E^{-1}_t A \bar{Z}^*_t = \bar{F}_t \bar{Z}^*_t,
	\end{align}
	where $E_t = (I + B R^{-1}_t B^T \bar P_{t+1} - C S^{-1}_t C^T \bar P_{t+1})$. Furthermore the co-state $\bar{\zeta}_t$ has the form $\bar{\zeta}_t = \bar{P}_t \bar{Z}^*_t$. Substituting this into the equilibrium adversarial control \eqref{eq:adv_cont} and observing that $\bar{\zeta}_t = - \bar{\zeta}^0_t$ we get
	\begin{align} 
	V^*_t = S^{-1}_t C^T \bar{\zeta}_{t+1} = S^{-1}_t C^T \bar{P}_t \bar{Z}^*_{t+1} = S^{-1}_t C^T \bar{P}_t \bar{F}_t \bar{Z}^*_t .
	\end{align}
	Hence we have shown that the equilibrium adversarial policy is linear in the equilibrium global MF, which has deterministic linear dynamics.
\end{proof}

\begin{proof}[Proof of Theorem \ref{thm:eq_loc_MF}]
	Let us define $\tilde \zeta_t = \sum_{j \in \Ns} \zeta^j_t/\lvert \Ns \rvert$, where $\zeta^j_t$ is the co-state of the $j$th agent. We can write the dynamics of these quantities as follows:
	\begin{align}
	Y^*_{t+1} & = A  Y^*_t - (B R^{-1}_t B^T - C S^{-1}_t C^T) \tilde \zeta_{t+1} + \tilde W_t,  \label{eq:lmf_dyn}\\
	\tilde \zeta_t & = A^T \tilde \zeta_{t+1} + (Q_t + \bar{Q}_t)  Y^*_t - \bar{Q}_t \bar{Z}^*_t - \tilde M^\zeta_t  \label{eq:lmf_cos}, \\
	\tilde \zeta_T & = (Q_T + \bar{Q}_T)  Y^*_T - \bar{Q}_T \bar{Z}^*_T, \nonumber
	\end{align}
	where $\tilde W_t = \sum_{j \in \Ns} W^j_t/\lvert \Ns \rvert$ and $\tilde M^\zeta_t$ is a Martingale difference sequence,
	$\tilde M^\zeta_t = A^T \tilde \zeta_{t+1} - A^T \EE \big[\tilde \zeta_{t+1} \mid \Fs^{ Y}_t \big]$. 
	Now we assume the form of the co-state $\tilde \zeta_t = \tilde P_t Y^*_t + \tilde s_t$. Substituting into the dynamics of the equilibrium local MF, we get
	\begin{align} \label{eq:lmf_dyn_inter}
	& Y^*_{t+1}    \\
	& = A  Y^*_t - (B R^{-1}_t B^T - C S^{-1}_t C^T) (\tilde P_{t+1} Y^*_{t+1} + \tilde s_{t+1}) + \tilde W_t \nonumber \\
	& = \tilde E^{-1}_t ( A  Y^*_t   - (B R^{-1}_t B^T - C S^{-1}_t C^T) \tilde s_{t+1} + \tilde W_t ) \nonumber
	\end{align}
	where $\tilde E_t = (I + B R^{-1} B^T \tilde P_{t+1} - C S^{-1}_t C^T \tilde P_{t+1})$. Now we substitute $\tilde \zeta_t = \tilde P_t Y^*_t + \tilde s_t$ into \eqref{eq:lmf_cos} to arrive at
	\begin{align}
	&\tilde P_t Y^*_t + \tilde s_t  = A^T \tilde P_{t+1} \EE [ Y^*_{t+1} \mid \Fs^Y_t] + A^T \tilde s_{t+1} \label{eq:lmf_inter} \\
	& \hspace{0.8cm} + (Q_t + \bar{Q}_t) Y^*_t - \bar{Q}_t  \bar{Z}^*_t \nonumber\\
	& = A^T \tilde P_{t+1} (\tilde E^{-1}_t ( A  Y^*_t   - (B R^{-1}_t B^T - C S^{-1}_t C^T) \tilde s_{t+1})  \nonumber \\
	& \hspace{0.8cm} + A^T \tilde s_{t+1} + (Q_t + \bar{Q}_t) Y^*_t - \bar{Q}_t \bar{Z}^*_t \nonumber, \\
	&\tilde P_T Y^*_T + \tilde s_T = (Q_T + \bar{Q}_T) Y^*_T - \bar{Q}_T  \bar{Z}^*_T \nonumber
	\end{align}
	Comparing coefficients of $Y^*_t$, we arrive at the Riccati equation:
	\begin{align*}
	\tilde P_t = A^T \tilde P_{t+1}\tilde E^{-1}_t A + Q_t + \bar{Q}_t, \hspace{0.2cm} \tilde P_T = Q_T + \bar{Q}_T
	\end{align*}
	From \eqref{eq:lmf_inter} we can compute the process $\tilde s_t$ in a backwards-in-time manner,
	\begin{align*}
	\tilde s_t = \tilde H_t \tilde s_{t+1} - \bar{Q}_t \bar{Z}^*_t, \hspace{0.2cm} \tilde s_T = - \bar{Q}_T \bar{Z}^*_T,
	\end{align*}
	where $\tilde H = A^T (I - \tilde P_{t+1} \tilde E^{-1}_t B R^{-1} B^T)$. Hence, the process $\tilde s_t$ can be computed as
	\begin{align*}
	\tilde s_t & = -\sum_{i=0}^{T-t} \prod_{j=1}^{i} \tilde H_{t+j} \bar{Q}_{t+i} \bar{Z}_{t+i} \nonumber \\
	& = -\sum_{i=0}^{T-t} \prod_{j=1}^{i} \tilde H_{t+j} \bar{Q}_{t+i}  \bar{F}_{t+j} \bar{Z}_t.
	\end{align*}
	The last equality is due to the fact that equilibrium global MF follows linear dynamics \eqref{eq:gmf_lin}. Substituting this expression into \eqref{eq:lmf_dyn_inter}, we get
	\begin{align}
	Y^*_{t+1} = \tilde F^1_t  Y^*_t   + \tilde F^2_t \bar{Z}^*_t + \tilde  E^{-1}_t \tilde W_t 
	\end{align}
	where 
	\begin{align*}
	\tilde F^1_t = & \tilde E^{-1} A, \\
	\tilde F^2_t = & \tilde E^{-1}_t (B R^{-1}_t B^T - C S^{-1}_t C^T))  \sum_{i=0}^{T-t} \prod_{j=1}^{i} \tilde H_{t+j} \bar{Q}_{t+i}  \bar{F}_{t+j} \nonumber
	\end{align*}
	Hence the equilibrium local MF has linear Gaussian dynamics driven by the equilibrium global MF.
\end{proof}

\end{document}